\documentclass{amsart}

\numberwithin{equation}{section}
\newtheorem{theorem}{Theorem}
\newtheorem{lemma}{Lemma}

\newtheorem{definition}{Definition}

\numberwithin{lemma}{section}
\numberwithin{proposition}{section}
\numberwithin{corollary}{section}
\numberwithin{theorem}{section}
\numberwithin{equation}{section}
\numberwithin{example}{section}

\begin{document}

\title
  {\bf The Cauchy problem of the Ward equation with mixed scattering data
  }

\author   { Derchyi Wu }
\maketitle

\begin{center}
INSTITUTE OF MATHEMATICS
\end{center}

\begin{center}
ACADEMIA SINICA
\end{center}

\begin{center}
TAIPEI, TAIWAN, R. O. C.
\end{center}

\begin{center}
mawudc@math.sinica.edu.tw
\end{center}

\maketitle

\section*{{Abstract}}
We solve the Cauchy problem of the Ward equation with both continuous and discrete scattering data.  

\vskip.25in

\section{Introduction }

\vskip.2in
The Ward equation (or the modified $2+1$ chiral model) 
\begin{equation}
\partial_t\left(J^{-1}\partial_t J\right)-\partial_x\left(J^{-1}\partial_xJ\right)-\partial_y\left(J^{-1}\partial_yJ\right)-\left[J^{-1}\partial_tJ,J^{-1}\partial_yJ\right]=0,\label{E:Ward}
\end{equation}
for $J:\mathbb{R}^{2,1}\to SU(n)$, $\partial_w=\partial/\partial w$, is obtained from a dimension reduction and a gauge fixing of the self-dual Yang-Mills equation on $\mathbb{R}^{2,2}$ \cite{DTU}, \cite{T06}. It is an integrable system which possesses the Lax pair \cite{W88}
\begin{equation}
\left[\lambda\partial_x-\partial_\xi-J^{-1}\partial_\xi J,\lambda\partial_\eta-\partial_x-J^{-1}\partial_x J\right]=0\label{E:Lax3}
\end{equation}
with $\xi=\frac{t+y}2$, $\eta=\frac{t-y}2$.  Note (\ref{E:Lax3}) implies that $J^{-1}\partial_\xi J=-\partial_x Q$, $J^{-1}\partial_x J=-\partial_\eta Q$. Then by a change of variables $(\eta,x,\xi)\rightarrow (x,y,t)$, (\ref{E:Lax3}) is equivalent to
\begin{eqnarray}
&&(\partial_y-\lambda\partial_x)\Psi(x,y,t,\lambda)=\left(\partial_x Q(x,y,t)\right)\Psi(x,y,t,\lambda), \label{E:Lax01}\\
&&(\partial_t-\lambda\partial_y)\Psi(x,y,t,\lambda)=\left(\partial_yQ(x,y,t)\right)\Psi(x,y,t,\lambda) \label{E:Lax02}
\end{eqnarray}
\cite{FI01}, and the Ward equation (\ref{E:Ward}) turns into:
\begin{equation}
\partial_x\partial_tQ=\partial^2_yQ+\left[\partial_yQ,\partial_xQ\right].\label{E:chiral}
\end{equation}

The construction of solitons, and the study of the scattering properties of solitons of the Ward equation have been studied intensively by solving the degenerated Riemann-Hilbert problem and studying the limiting method 
\cite{W88}, \cite{W95}, \cite{I96}, \cite{An97}, \cite{An98}, \cite{IZ98}. In particular, Dai and Terng gave an explicit construction of all solitons of the Ward equation by establishing a theory of Backlund transformation \cite{DT07}.

On the other hand, 
 Villarroel \cite{V90}, Fokas and Ioannidou \cite{FI01}, Dai, Terng and Uhlenbeck \cite{DTU} investigate the scattering and inverse scattering problem and solve the Cauchy problem of the Ward equation if the initial potential is sufficiently small.  
Under the small data condition, the eigenfunctions $\Psi$ 
 possesses continuous scattering data only and therefore the solutions for the Ward equation do not include the solitons in previous study. 
 
Observing the similarity between Lax pairs of the Ward equation and the AKNS system \cite{BC84}, \cite{BC85} and transforming the spectral equation (\ref{E:Lax01}) into a $\bar\partial$-equation, the author  uses the Beals-Coifman scheme, the $L_2$-theory of the Cauchy integral operator and estimation of the singular integral to solve the scattering and inverse scattering problem and the Cauchy problem of the Ward equation if the initial potential possesses purely continuous scattering data \cite{Wu}.

The purpose of the present paper is to apply the above existence theory and the Backlund transformation theory to solve the the Cauchy problem of the Ward equation with mixed scattering data. The scheme we will follow was developed by \cite{TU00}, \cite{DTU} to construct AKNS flows and Ward equations. 
More precisely, we will generalize the results of  \cite{V90}, \cite{FI01}, \cite{DTU}, \cite{Wu} by the following theorem:

\begin{theorem}\label{T:cauchy2}
If $Q_0\in \mathbb P_{\infty,k,1}$, $k\ge 7$ (see Definition \ref{D:L11}), and the poles of $\Psi_0$  (including the multiplicity) are finite and contained in $\mathbb C\backslash\mathbb R$, then the Cauchy problem of the Ward equation (\ref{E:chiral}) with initial condition $Q(x,y,0)=Q_0(x,y)$
admits a smooth global solution  satisfying: 
for $i+j+h\le k-4$, $i^2+j^2>0$,
\begin{eqnarray*}
&&\partial_x Q(x,y,t)\in su(n),\\
&&\partial_x^i\partial_y^j\partial_t^hQ,\,\partial_tQ,\,Q\in L_\infty,\\
&&\partial_x^i\partial_y^j\partial_t^hQ,\,\partial_tQ,\,Q\to 0, \,{\textit as }\,\,x,\,y,\,t\to\infty.
\end{eqnarray*}
\end{theorem}

\section{The Cauchy problem with purely continuous scattering data}\label{E:CauchyContinuous}
We summarize the results of \cite{Wu} which are necessary in proving Theorem \ref{T:cauchy2}. First of all, define the functional spaces
\begin{definition}\label{D:L11}
\begin{eqnarray*}
&&\mathbb P_{\infty,k_1,k_2}\\
=&&\{q_x(x,y):\mathbb{R}\times\mathbb{R}\to su(n)|\,\\
 &&\begin{array}{lll}
 |\xi^iy^s\widehat q|_{L_1(d\xi dy)},&||\xi^h\widehat q(\xi,y)|_{L_1(y)}|_{ L_2(d\xi)},&{}\\
|\partial_x^j\partial_y^lq|_{L_\infty},&\sup_y|\partial_x^j\partial_y^lq|_{L_1(dx)},&|\partial_x^j\partial_y^lq|_{L_1(dxdy)}<\infty\end{array}\\
&&\textit{ for $1\le i\le \max\{5,k_1\},\,\,\,\, 0\le j,\, l\le \max\{5,k_1\},\,\,1\le h\le k_1$, $0\le s\le k_2$}\,\}.\\&&\mathbb{DH}^k\\
=&&\{f\,| \partial_x^if(x,y)
\textit{ are uniformly bounded in } L_2(\mathbb R, dx), \,\,0\le i\le k.\}
\end{eqnarray*}
 \end{definition}

\begin{theorem}\label{T:LNexistence}
Let $Q\in \mathbb P_{\infty,k,0}$, $k\ge 2$. Then there is a bounded set $Z\subset \mathbb{C}\backslash \mathbb{R}$ such that 
\begin{itemize}
	\item $Z\cap \left(\mathbb C\backslash\mathbb R\right)$ is discrete in $\mathbb C\backslash\mathbb R$;
	\item For  $\lambda\in \mathbb{C}\backslash \left(\mathbb{R}\cup Z\right)$, the problem 
(\ref{E:Lax01}) has a unique solution $\Psi$ and $\Psi-1\in \mathbb {DH}^k$;
  \item For $(x,y)\in \mathbb{R}\times \mathbb{R}$, the eigenfunction $\Psi(x,y,\cdot)$ is meromorphic in $\lambda\in\mathbb{C}\backslash \mathbb{R}$ with poles precisely at the points of $Z$;
  \item $\Psi(x,y,\lambda)$ satisfies:
  \begin{eqnarray}
&&\lim_{|x|\to\infty}\Psi(\cdot,y,\lambda)=1,\,\,
\lim_{|y|\to\infty}\Psi(x,\cdot,\lambda)=1,\,\,\textit{ for  $\lambda\in \mathbb{C}\backslash \left(\mathbb{R}\cup Z\right)$}, \label{E:bdry}\\
&&\textit{$\Psi(x,y,\cdot)$ tends to $1$ uniformly as $|\lambda|\to\infty$}.\label{E:bdry''}
\end{eqnarray}
\item $\Psi(x,0,\lambda)$ satisfies:
  \begin{eqnarray}
&& \Psi-1\textit{ are uniformly bounded in ${DH}^k$ for $\lambda\in \mathbb{C}\backslash \left(\mathbb{R} \cup_{\lambda_j\in Z} D_\epsilon(\lambda_j)\right)$. }\label{E:bdry'}\\
&&\textit{For any $z_j\in\mathbb C\backslash\mathbb R$, fixing $\epsilon_k$ for $\forall k\ne j$ and letting $\epsilon_j\to 0$, these  }\nonumber\\
&&\textit{$L_2(dx)$-norms increase as $ {C_j}{\epsilon_j^{-h_j}}$ with uniform constants $C_j$, $h_j>0$;}\nonumber\\
&& \textit{$\Psi-1\to 0$  in ${DH}^{k}$ as $\lambda\to\infty$.}
\label{E:bdry'''}
\end{eqnarray}
  \end{itemize}
  Where $\epsilon>0$ is any given constant, $D_\epsilon(\lambda_j)$ is the disk of radius $\epsilon$ centered at $\lambda_j$. Finally, for $\lambda\notin Z$,
 \begin{eqnarray}
  &&\det \Psi(x,y,\lambda)\equiv 1,\,\,\textit{ }\label{E:det}\\
&&\Psi(x,y,\lambda)\Psi(x,y,\bar\lambda)^*=I.\label{E:reality2}
 \end{eqnarray} 
\end{theorem}
\begin{proof} See Theorem 1.1, Corollary 4.1 and Lemma 4.14 in \cite{Wu}. \end{proof}

\begin{theorem}\label{T:CSDsum}
For $Q\in {\mathbb P}_{\infty,k,1}$, $k\ge 7$, if $\Psi_\pm$ exist, then there exists uniquely a function $v(x, y, \lambda)$ such that 
\[
\Psi_{+}(x,y,\lambda)=\Psi_-(x,y,\lambda)v(x,y, \lambda), \,\, \lambda\in R.
\]Moreover, 
$v$ satisfies the algebraic constraints:
\begin{eqnarray}
&&\det \,(v)\equiv 1, \label{E:real1'}\\
&& v= v^*>0,\label{E:real2'}
\end{eqnarray}
and the analytic  constraints: for $i+j\le k-4$, 
\begin{eqnarray}
&&\mathcal L_\lambda v=0,\,\,v(x,y,\lambda)=v(x+\lambda y, \lambda)\textit{ for }\forall x,\,y\in\mathbb R,\label{E:0ana15'}\\
&&\partial_x^i\partial_y^j\left(v-1\right)\textit{ are uniformly bounded in $ L_\infty\cap L_2(\mathbb{R},d\lambda)\cap L_1(\mathbb{R},d\lambda)$};\label{E:0anal2'}\\ 
&&\partial_x^i\partial_y^j\left(v-1\right)\to 0\textit{ uniformly in  $L_\infty\cap L_2(\mathbb{R},d\lambda)\cap L_1(\mathbb{R},d\lambda)$   }\label{E:0ana14'}\\
&&\textit{as $|x|$ or $|y|\to\infty$;}\nonumber\\
&&\partial_\lambda v\, \textsl{ are in $  L_2(\mathbb{R},d\lambda)$ and the norms depend continuously on $x$, $y$.}\label{E:0anal3'}
\end{eqnarray}
Where 
$
\mathcal L_\lambda=\partial_y-\lambda\partial_x$. 
\end{theorem}  
\begin{proof} Note it is of no harm to replace the condition $Z=Z(\Psi)=\phi$ by the existence of $\Psi_\pm$ in the proof of Theorem 1.2 in \cite{Wu}.
\end{proof}

\begin{theorem}\label{T:invexistence}
Suppose $v(x, y, \lambda)$ satisfies (\ref{E:real1'}), (\ref{E:real2'}), and (\ref{E:0anal2'})-(\ref{E:0anal3'}), $k\ge 7$. Then  
there exists a unique solution $\Psi(x,y,\cdot)$ for the Riemann-Hilbert problem $(\lambda\in\mathbb{R}, v(x, y, \lambda))$ such that
\begin{equation}
\Psi-1,\, \partial_x\Psi,\,\partial_y\Psi \textit{ are uniformly bounded in $L_2(\mathbb R, d\lambda)$.} \label{E:invbdry''}
\end{equation}
In addition, for each fixed $\lambda\notin\mathbb R$, and  $i+j\le k-4$,
\begin{gather}
\partial_x^i\partial_y^j\Psi\in L_\infty(dxdy),\label{E:invbdry'}\\
\partial_x^i\partial_y^j\left(\Psi-1\right)\to 0  \textit{ in $L_\infty(dxdy)$, as $x$ or $y\to \infty$}.\label{E:invbdry}
\end{gather}
Moreover,
\begin{eqnarray}
&&\det \Psi(x,y,\lambda)\equiv 1,\label{E:invsym0}\\
&&\Psi(x,y,t,\lambda)\Psi(x,y,t,\bar\lambda)^*\equiv 1.\label{E:invsym1}
\end{eqnarray}
\end{theorem}
\begin{proof} See Theorem 1.3 and Lemma 6.6 in \cite{Wu}. Note (\ref{E:0ana15'}) is not used in the proof of Theorem 1.3 in \cite{Wu}. It is only used in justifying the existence of $Q(x,y)$ (see the proof of Theorem 1.4 in \cite{Wu}). \end{proof}

\begin{theorem}\label{T:cauchy}
If $Q_0\in \mathbb P_{\infty,k,1}$, $k\ge 7$, and there are no poles of the eigenfunction $\Psi_0$ of $Q_0$, then the Cauchy problem of the Ward equation (\ref{E:chiral}) with initial condition $Q(x,y,0)=Q_0(x,y)$
admits a smooth global solution satisfying: for $i+j+h\le k-4$, $i^2+j^2>0$,
\begin{eqnarray*}
&&\partial_x Q(x,y,t)\in su(n),\\
&&\partial_x^i\partial_y^j\partial_t^hQ,\,\partial_tQ,\,Q\in L_\infty,\\
&&\partial_x^i\partial_y^j\partial_t^hQ,\,\partial_tQ,\,Q\to 0, \,{\textit as }\,\,x,\,y,\,t\to\infty.
\end{eqnarray*}
 
\end{theorem}
\section{The Cauchy problem with purely discrete scattering data }\label{E:CauchyDiscrete}
We review  
results of \cite{DT07} which are necessary for the proof of Theorem \ref{T:cauchy2}.
\begin{definition}\label{D:loopgroup}
\begin{eqnarray*}
&&\textit{$\mathcal{D}=\{f: \mathbb{C}\to M_n( \mathbb{C})$ satisfies the following conditions:}\\
&&\hskip.5in\bullet\,\, f(\bar\lambda)^*f(\lambda)=1, \,\,\lim_{\lambda\to\infty}f=I.\\
&&\hskip.5in\bullet\,\, \textit{$f$ is meromorphic in $\mathbb{C}\backslash \mathbb{R}$ with possible finitely many poles.}\\
&&\hskip.5in\bullet\,\, \textit{$f\pm$ exist.}\};\\
&&\mathcal{D}_c=\{f\in \mathcal{D}\textit{ such that $f$ is holomorphic in $\mathbb{C}\backslash \mathbb{R}$ and $f_\pm$ exist.}\};\\
&&\mathcal{D}_r=\{f\in \mathcal{D}\textit{ such that $f$ is rational.}\}.\\
&&\textsl{$Z(f)=$  the set of poles of $f$, for $f\in\mathcal D$.}
\end{eqnarray*}\end{definition}
\begin{theorem}\label{T:uhlenbeck}
For $\psi\in \mathcal{D}_r$ and $Z(f)\subset\mathbb C^+$,  we can factorize $\psi$ as a product of $f_1\cdots f_p$. Where $f_i\in\mathcal D_r$, $Z(f_i)=\{z_i\}\subset\mathbb C^+$ and $z_i\ne z_j$ if $i\ne j$.
\end{theorem}
\begin{proof} See Theorem 5.4 in \cite{U89} and Corollary 4.5 in \cite{DT07}. 
\end{proof}

\begin{definition}\label{D:fundamental} Let $A_0(x,y)$, $A(x,y,t)$, and $B(x,y,t)\in M_n( \mathbb{C})$.

\noindent\begin{itemize}
	\item $\eta_0(x,y)$ is called a fundamental solution of the system
\begin{eqnarray}
&&(\partial_y-z\partial_x)f=A_0(x,y)f,\label{E:fundsimple}
\end{eqnarray}
if for $\forall V(x,y)$ satisfies (\ref{E:fundsimple}), there exists $ H(x+zy)$ such that \[V(x,y)=\eta_0(x,y) H(x+zy).\]
\item $\eta(x,y,t)$ is called a fundamental solution of
\begin{eqnarray}
&&(\partial_y-z\partial_x)f=A(x,y,t)f,\label{E:fund1}\\
&&(\partial_t-z\partial_y)f=B(x,y,t)f,\label{E:fund2}
\end{eqnarray}
if for $\forall V(x,y,t)$ satisfies (\ref{E:fund1}), (\ref{E:fund2}), there exists $  H(x+zy+z^2t)$ such that \[V(x,y,t)=\eta(x,y,t) H(x+zy+z^2t).\]
\end{itemize}
\end{definition}

The next goal is to solve the following Cauchy problem (\ref{E:Laxini})-(\ref{E:bdc1}).
\begin{theorem}\label{T:cauchydiscrete}
Suppose there exist $A_0(x,y)$, $\Psi_0(x,y,\lambda)$ such that  
\begin{equation}
(\partial_y-\lambda\partial_x)\Psi_0(x,y,\lambda)=A_0(x,y)\Psi_0(x,y,\lambda), \label{E:Laxini}
\end{equation}
and
\begin{gather}
\textit{$\Psi_0(x,y,\cdot)\in\mathcal D_r$, $Z(\Psi_0)=\{z\}$ with multiplicity $k$,}\label{E:rational}\\
\textit{$\partial_x^i\partial_y^j\Psi_0\in L_\infty$ for $i+j\le N$, and $\lambda\in\mathbb C\backslash\mathbb \{z\}$ fixed, }\label{E:diff}\\
\textit{$\lim_{x^2+y^2\to\infty}\partial_x^i\partial_y^j\left(\Psi_0- C_1\right)= 0$ for $i+j\le N$, and $\lambda\in\mathbb C\backslash\mathbb \{z\}$ fixed. }\label{E:bdc}
\end{gather}
Where $C_1$ is independent of 	 $x$, $y$. 
Then  there exist $\Psi(x,y,t,\lambda)$, $A(x,y,t)$, $B(x,y,t)$ such that 
\begin{eqnarray}
&&(\partial_y-\lambda\partial_x)\Psi(x,y,t,\lambda)=A(x,y,t)\Psi(x,y,t,\lambda), \label{E:Lax03}\\
&&(\partial_t-\lambda\partial_y)\Psi(x,y,t,\lambda)=B(x,y,t)\Psi(x,y,t,\lambda), \label{E:Lax04}\\
&&\Psi(x,y,0,\lambda)=\Psi_0(x,y,\lambda),\,\,\,A(x,y,0)=A_0(x,y),\label{E:initial}
\end{eqnarray}
and
\begin{gather}
\textit{$\Psi(x,y,t,\cdot)\in\mathcal D_r$, $Z(\Psi)=\{z\}$ with multiplicity $k$,}\label{E:rational1}\\
\textit{$\partial_x^i\partial_y^j\partial_t^h\Psi\in L_\infty$ for $i+j+h\le N$,  and $\lambda\in\mathbb C\backslash\mathbb \{z\}$ fixed, }\label{E:diff1}\\
\textit{$\lim_{x^2+y^2+t^2\to\infty}\partial_x^i\partial_y^j\partial_t^h\left(\Psi_0- C_1\right)= 0$ for $i+j+h\le N$, and $\lambda\in\mathbb C\backslash\mathbb \{z\}$ fixed.}\label{E:bdc1}
\end{gather}
\end{theorem}
\begin{proof} This theorem is indeed a rephrase of a very decent result (Theorem 8.6) in \cite{DT07}. For convenience,  the proof is sketched here. The theorem will be proved  by simultaneously establishing (\ref{E:Lax03})-(\ref{E:bdc1}),  and the following statements (\ref{E:fund3})-(\ref{E:fund6}): 
\begin{eqnarray}
	&&  \textsl{There exists a fundamental solution, $\eta_0(x,y)$, of (\ref{E:fundsimple});}\label{E:fund3}\\
	&&  \textsl{There exists a fundamental solution, $\eta(x,y,t)$, of (\ref{E:fund1}) and (\ref{E:fund2});}\label{E:fund4}\\
	&&\eta(x,y,0)=\eta_0(x,y);\label{E:fund5}\\
  &&\textit{$\partial_x^i\partial_y^j\partial_t^h\eta\in L_\infty$ for $i+j+h\le N$, and $\lambda\in\mathbb C\backslash\mathbb \{z\}$ fixed; }\label{E:diff2}\\	
	&&\textit{$\lim_{x^2+y^2+t^2\to\infty}\partial_x^i\partial_y^j\partial_t^h\left(\eta - C_2\right)=0$ for $i+j+h\le N$, and $\lambda\in\mathbb C\backslash\mathbb \{z\}$ fixed.}\label{E:fund6}
	\end{eqnarray}
 Where $C_2$ is independent of $x$, $y$. 

Suppose $k=1$. By the results of \cite{U89}, \cite{W88}, \cite{W89}, one has
\begin{eqnarray}
\Psi_0(x,y,\lambda)&=&1+\frac{z-\bar z}{\lambda-z}\pi_0^\bot(x,y),\label{E:p0}\\
A_0(x,y)&=&(z-\bar z)\partial_x\pi_0.\label{E:A0}
\end{eqnarray}
Where $\pi_0(x,y):\mathbb R\times\mathbb R\to\mathbb C^n $, $\pi_0^*=\pi_0$, $\pi_0^2=\pi_0$, $\mathrm {Im }\pi_0$ is spanned by   columns of 
\begin{equation}
V(x+zy),\qquad V:\mathbb C\to M_{n\times r} \textit{ is holomorphic except at $\{p_1,\cdots,p_s\}$},\label{E:V}
\end{equation}
and $M_{n\times r}$ denotes the space of rank $r$ complex $n\times r$ matrices. Note  by (\ref{E:diff}), one can find another $\tilde V(x+zy)$, $\tilde V:\mathbb C\to M_{n\times r}$ is holomorphic in the neighborhood of  $p_1$, $\cdots$, $p_s$, and $\mathrm {Im }\pi_0$ is spanned by   columns of $\tilde V$. Hence  (\ref{E:p0})-(\ref{E:V}) hold for $\forall x,\,y$.

Now let us define 
\begin{eqnarray*}
&&\Psi(x,y,t,\lambda)=1+\frac{z-\bar z}{\lambda-z}\pi^\bot(x,y,t),\\
&&\pi(x,y,t):\mathbb R\times\mathbb R\times\mathbb R\to\mathbb C^n , \,\,\pi^*=\pi,\,\,\pi^2=\pi,\\
&& \textit{$\mathrm{Im }\pi$ is spanned by columns of $V(x+zy+z^2t)$,}
\end{eqnarray*}
then  (\ref{E:Lax03})-(\ref{E:bdc1}) are satisfied with 
\begin{equation}
A(x,y,t)= (z-\bar z)\partial_x\pi,\qquad B(x,y,t)= (z-\bar z)\partial_y\pi.\label{E:AB}
\end{equation}

To prove (\ref{E:fund3})-(\ref{E:fund6}), by a Dai and Terng's construction (Theorem 8.6 in \cite{DT07}), one can find a fundamental solution $\eta(x,y,t)$ of (\ref{E:fund1}) and (\ref{E:fund2}) by defining 
\begin{gather*}
\eta(x,y,t)=\left[\xi_1\,\cdots\,\xi_{n_1}\,u_{n_1+1}\,\cdots\, u_n\right],\\
\xi_j=\lim_{\epsilon\to 0}\Psi(x,y,t,z+\epsilon)(a_j(x+(z+\epsilon)y+(z+\epsilon)^2t)),\label{E:2}\\
u_j=\lim_{\epsilon\to 0}\Psi(x,y,t,z+\epsilon)(\epsilon b_j(x+(z+\epsilon)y+(z+\epsilon)^2t)),\label{E:1}\\
\textsl{$a_1$, $\cdots$, $a_{n_1}$ span $\mathrm{Im }\pi$, $\pi^\bot(b_{n_1+1})$, $\cdots$, $\pi^\bot(b_n)$ span $\mathrm{Im }\pi^\bot$.}\label{E:3}
\end{gather*}
Defining $\eta_0(x,y)$ by (\ref{E:fund5}) and using (\ref{E:Lax03})-(\ref{E:initial}), (\ref{E:diff1}), (\ref{E:bdc1}), the invertibility of $\eta$, one has (\ref{E:fund3}), (\ref{E:diff2}), and (\ref{E:fund6}).


Now suppose (\ref{E:Lax03})-(\ref{E:fund6}) hold  for $k$ and assume that $Z(\Psi_0)=\{z\}$ with multiplicity $k+1$. Using the minimal factorization property of elements in $\mathcal D_r$ (see Section 4 and Theorem 8.1 in \cite{DT07}), one would obtain 
\begin{eqnarray}
&&\Psi_0(x,y,\lambda)=g_{z,\pi_{k+1,0}}\Psi_{k,0}(x,y,\lambda),\label{E:F1}\\
&&g_{z,{k+1,0}}(x,y,\lambda)=1+\frac{z-\bar z}{\lambda-z}\pi_{k+1,0}^\bot(x,y), \label{E:F2}\\
&&(\partial_y-\lambda\partial_x) \Psi_{k,0}= A_{k,0}(x,y)\Psi_{k,0}, \label{E:F3}\end{eqnarray} and
\begin{gather}
\textsl{$\Psi_{k,0}\in \mathcal D_r$, $Z(\Psi_{k,0})=\{z\}$ with multiplicity $k$},\label{E:F4}\\
\textit{$\partial_x^i\partial_y^jg_{z,\pi_{k+1,0}}, \,\partial_x^i\partial_y^j\Psi_{k,0}\in L_\infty$ for $i+j\le N$, and $\lambda\in\mathbb C\backslash\mathbb \{z\}$ fixed; }\label{E:diff2'}\\	
	\textit{$\partial_x^i\partial_y^j\left(g_{z,\pi_{k+1,0}}- C_3\right)\to 0,\,\partial_x^i\partial_y^j\left(\Psi_{k,0}- C_4\right)\to 0$ as $|x|$, or $|y|\to\infty$, }\label{E:fund6'}\\
	\textit{and $\lambda\in\mathbb C\backslash\mathbb \{z\}$ fixed.}\nonumber
\end{gather}Where $C_3$, $C_4$ are independent of $x$, $y$. Conditions (\ref{E:F1})-(\ref{E:fund6'}), (\ref{E:Laxini}), and Proposition 7.1, 7.4 in \cite{DT07} conclude that  
there exist
\begin{gather}
	\textit{$ \tau_{1,0}(x,y),\cdots,\tau_{n_{k+1},0}(x,y)$ span $\mathrm{Im }\,\pi_{k+1,0}$,}\label{E:tau}\\
	(\partial_y-z\partial_x)\tau_{j,0}= A_{k,0}(x,y)\tau_{j,0},\label{E:tildeA}\\
	\textit{$\partial_x^i\partial_y^h\tau_{j,0}\in L_\infty$ for $i+h\le N$, and $\lambda\in\mathbb C\backslash\mathbb \{z\}$ fixed, }\label{E:diff2''}\\	
	\textit{$\partial_x^i\partial_y^h\left(\tau_{j,0}- c_j\right)\to 0$ as $|x|$, or $|y|\to\infty$, and $\lambda\in\mathbb C\backslash\mathbb \{z\}$ fixed.}\label{E:fund6''}
\end{gather}
Where $c_j$ are independent of $x$, $y$. Besides, by (\ref{E:F3})-(\ref{E:fund6'}), and the induction hypothesis, there exists a fundamental solution  $\eta_{k,0}(x,y)$ of (\ref{E:tildeA}).
Hence there are $H_1$, $\cdots$, $H_{n_{k+1}}:\mathcal C\to M_{n\times 1}$ such that
\begin{gather}
\tau_{j,0}(x,y)=\eta_{k,0}(x,y)H_j(x+zy),\label{E:exp}\\
\textit{$\partial_x^i\partial_y^hH_j\in L_\infty$ for $i+h\le N$, and $\lambda\in\mathbb C\backslash\mathbb \{z\}$ fixed, }\label{E:diff2'''}\\	
	\textit{$\partial_x^i\partial_y^h\left(H_j- c_j'\right)\to 0$ as $|x|$, or $|y|\to\infty$, and $\lambda\in\mathbb C\backslash\mathbb \{z\}$ fixed.}\label{E:fund6'''}
\end{gather} 
Where $c_j'$ are independent of $x$, $y$. Moreover, by induction, one can find $\Psi_k(x,y,t,\lambda)$ solution of (\ref{E:Lax03})-(\ref{E:bdc1}),  $\eta_k$ a solution of (\ref{E:fund4})-(\ref{E:fund6}) with $\Psi_0$, $A_0$, $A$, $B$, $\eta_0$, $\eta$ replaced by $\Psi_{k,0}$, $ A_{k,0}$, $A_k$, $B_k$, $\eta_{k,0}$, $\eta_k$. So by (\ref{E:exp}), (\ref{E:tau}), and (\ref{E:F1}),  one can extend $\tau_{j,0}(x,y)$, $\pi_{k+1,0}$, $\Psi$ by 
\begin{eqnarray*}
&&\tau_j(x,y,t)=\eta_k(x,y,t)H(x+zy+z^2t),\\
&& \pi_{k+1}(x,y,t),\textsl{ a Hermitian projection on the space spanned by $\tau_j$},\\
&&\Psi(x,y,t,\lambda)=g_{z,\pi_{k+1}}\Psi_k(x,y,t,\lambda).
\end{eqnarray*}
Therefore,  we can apply Lemma 8.2, and Proposition 7.1 in \cite{DT07} to prove (\ref{E:Lax03})-(\ref{E:bdc1}) for the case of $k+1$. 

To prove (\ref{E:fund3})-(\ref{E:fund6}) for the case of $k+1$, by a Dai and Terng's construction,  one can find a fundamental solution $\eta(x,y,t)$ of (\ref{E:fund1}) and (\ref{E:fund2}) by defining 
\begin{gather*}
\eta(x,y,t)=\left[\xi_1\,\cdots\,\xi_{n_{k+1}}\,\zeta_{1}\,\cdots\, \zeta_{n-n_{k+1}}\right],\\
\xi_j=\lim_{\epsilon\to 0}g_{z,\pi_{k+1}}(x,y,t,z+\epsilon)\left(a_j(x+(z+\epsilon)y+(z+\epsilon)^2t)\right),\\
\textsl{$a_1$, $\cdots$, $a_{n_{k+1}}$ span $\mathrm{Im }\pi_{k+1}$, }
\textit{$\zeta_1$, $\cdots$, $\zeta_{n-n_{k+1}}$ span $\pi^\bot_{k+1}\eta_k$.}
\end{gather*}
Defining $\eta_0(x,y)$ by (\ref{E:fund5}) and using (\ref{E:Lax03})-(\ref{E:initial}), (\ref{E:diff1}), (\ref{E:bdc1}), the invertibility of $\eta$, one has (\ref{E:fund3}), (\ref{E:diff2}), and (\ref{E:fund6}).
\end{proof}

\begin{theorem}\label{T:backlund}
Suppose both $f(x,y,t,\lambda)$ and $g(x,y,t,\lambda)$ satisfy (\ref{E:Lax03}), (\ref{E:Lax04}) with different $A$, $B$. If $Z\left(g\right)=\{z\}$ with multiplicity of $k$, and $f$ is holomorphic and non-degenerate at $z$, $\bar z$. Then there exist unique $\tilde f$ and $\tilde g$ such that 
\begin{itemize}
\item $Z\left(\tilde g\right)=\{z\}$ with multiplicity of $k$, $\tilde f$ is holomorphic and non-degenerate at $z$, $\bar z$.
	\item $\Psi=\tilde fg=\tilde gf$ satisfies (\ref{E:Lax03}), (\ref{E:Lax04}) with new $A$, $B$.
\end{itemize}
Moreover, if  for $i+j+h\le N$, and $\lambda\in \mathbb C\backslash \left(\{z\}\cup Z (f)\right)$ fixed,
\begin{gather}
\textit{$\partial_x^i\partial_y^j\partial_t^h f,\,\,\partial_x^i\partial_y^j\partial_t^h g\in L_\infty$ }\label{E:diff3}\\
\textit{$\lim_{x^2+y^2+t^2\to\infty}\partial_x^i\partial_y^j\partial_t^h\left(f- C_1\right)=\lim_{x^2+y^2+t^2\to\infty}\partial_x^i\partial_y^j\partial_t^h\left(g- C_2\right)=0$ }\label{E:bdc3'}
\end{gather}
then for $i+j+h\le N$, and $\lambda\in \mathbb C\backslash \left(\{z\}\cup Z (f)\right)$ fixed,
\begin{gather}
\textit{$\partial_x^i\partial_y^j\partial_t^h \Psi\in L_\infty$}\label{E:diff4}\\
\textit{$\lim_{x^2+y^2+t^2\to\infty}\partial_x^i\partial_y^j\partial_t^h\left(\Psi- C_3\right)=0$.}\label{E:boundary}
\end{gather}
Where $C_1$, $C_2$, $C_3$ are independent of $x$, $y$, $t$. 
\end{theorem}
\begin{proof}
See Theorem 6.1 in \cite{DT07}. Note the regularity property (\ref{E:diff4}) and the asymptotic property (\ref{E:boundary}) are proved by using the property "` each component of the minimal factorization  of $g$ satisfies (\ref{E:diff3}), and (\ref{E:bdc3'})"'. 
\end{proof}

The same argument yields
\begin{theorem}\label{T:backlundinitial}
Suppose both $f(x,y,\lambda)$ and $g(x,y,\lambda)$ satisfy (\ref{E:Laxini}) with different $A_0$. If $Z\left(f\right)=\{z\}$ with multiplicity of $k$, and $g$ is holomorphic and non-degenerate at $z$, $\bar z$. Then there exist unique $\tilde f$ and $\tilde g$ such that 
\begin{itemize}
\item $Z\left(\tilde f\right)=\{z\}$ with multiplicity of $k$, $\tilde g$ is holomorphic and non-degenerate at $z$, $\bar z$.
	\item $\Psi=\tilde fg=\tilde gf$ satisfies (\ref{E:Laxini}) with a new $A_0$.
\end{itemize}
\end{theorem}

\section{The Cauchy problem: Mixed scattering data} \label{S:cauchy2}

\begin{lemma}\label{L:decom}
For $Q\in {\mathbb P}_{\infty,k,1}$, $k\ge 7$, if $Z=Z(\Psi)$ is finite and $Z\cap\mathbb R=\phi$, then there exist uniquely $f$, $\tilde f\in \mathcal D_c$, $g$, $\tilde g\in\mathcal D_r$, such that 
\begin{equation}
\Psi=\tilde fg=\tilde g f.\label{E:decom}
\end{equation}
Moreover, for $i+j\le k-4$, and $\lambda\in \mathbb C\backslash \left(\mathbb R\cup Z\right)$ fixed,
\begin{gather}
\textit{$\partial_x^i\partial_y^j f,\,\,\partial_x^i\partial_y^j g$, $\partial_x^i\partial_y^j\tilde  f,\,\,\partial_x^i\partial_y^j \tilde g\in L_\infty$;}\label{E:diff5}\\
\textit{$\partial_x^i\partial_y^j(f-1),\,\,\partial_x^i\partial_y^j(g-1),\,\,\partial_x^i\partial_y^j(\tilde f-1),\,\,\partial_x^i\partial_y^j(\tilde g-1)\to 0$, as $x,\,y\to\infty$;}\label{E:boundary5}\\
\det f=\det g=\det\tilde f=\det\tilde g=1.\label{E:determinant}
\end{gather}
\end{lemma}  
\begin{proof} Since $Z$ is finite. Theorem \ref{T:LNexistence} implies the existence of $\Psi_\pm$. Hence Theorem \ref{T:CSDsum} implies that: $\exists\,v_c(x, y, \lambda)$, $\lambda\in\mathbb{R}$ such that
\begin{gather}
\Psi_+=\Psi_-v_c;\label{E:dual}\\
\textsl{$v_c $ satisfies the algebraic and   analytic constraints (\ref{E:real1'})-(\ref{E:0anal3'}).  } \nonumber
\end{gather}
By Theorem \ref{T:invexistence}, we can find a function $f(x,y,\lambda)$, 
such that $f$ satisfies (\ref{E:diff5})-(\ref{E:determinant}), $f(x,y,\cdot)\in\mathcal D_c$, and
\begin{equation}
f_+=f_-v_c.\label{E:riemann}
\end{equation} 
Hence if we define $\tilde g=\Psi f^{-1}$, then Theorem \ref{T:LNexistence}, $f(x,y,\cdot)\in\mathcal D_c$, (\ref{E:dual}), (\ref{E:riemann}) and $Z$ is finite imply
\[\textit{$\tilde g(x,y,\lambda)\in \mathcal D_r$ and $\tilde g(x,y,\cdot)$ satisfies (\ref{E:diff5})-(\ref{E:determinant}).}
\]
Hence we prove the existence of $f$ and $\tilde g$. 

Besides, using a similar argument as that in the proof of Theorem \ref{T:CSDsum},  there exists $\tilde v_c(x, y, \lambda)$, $\lambda\in\mathbb{R}$ such that
\begin{gather}
\Psi_+=\tilde v_c\Psi_-;\label{E:tildedual}\\
\textsl{$\tilde v_c $ satisfies the algebraic and   analytic constraints (\ref{E:real1'})-(\ref{E:0anal3'}) except (\ref{E:0ana15'}).  } \nonumber
\end{gather}
By a similar argument as that in the proof of Theorem \ref{T:invexistence}, there exists a function $\tilde f(x,y,\lambda)$, 
such that $\tilde f$ satisfies (\ref{E:diff5}), (\ref{E:boundary5}), $\tilde f(x,y,\cdot)\in\mathcal D_c$,  and
\begin{equation}
\tilde f_+=\tilde v_c\tilde f_-.\label{E:tilderiemann}
\end{equation} 
Hence if we define $g=\tilde f^{-1}\Psi $, then Theorem \ref{T:LNexistence}, $\tilde f(x,y,\cdot)\in\mathcal D_c$, (\ref{E:tildedual}), (\ref{E:tilderiemann}) and $Z$ is finite imply $g(x,y,\cdot)\in\mathcal D_r$ and $ g$ satisfies (\ref{E:diff5})-(\ref{E:determinant}).

Suppose $\Psi=\tilde g_1 f_1=\tilde gf$  where $f_1$, $\tilde g_1$ satisfy the statement of the theorem. Hence $\tilde g^{-1}\tilde g_1=f f_1^{-1}$. The right hand side is holomorphic in $\mathbb{C}\backslash\mathbb{R}$, the left hand side is continuous across the real line and tends to $1$ at infinity. Thus the Liouville theorem yields $\tilde g^{-1}\tilde g_1=f f_1^{-1}=1$ and we prove the uniqueness of $f$ and $\tilde g$. The uniqueness of $\tilde f$ and $g$ can be obtained by analogy. 
\end{proof}

\begin{lemma}\label{L:tail}
Suppose $\Psi$ satisfies (\ref{E:Laxini}), $\Psi(x,y,\cdot)\in\mathcal D$, and $Z(\Psi)\subset\mathbb C^+$. If 
\begin{gather}
\Psi=\psi_1 g_1=g_2\psi_2,\label{E:fact1}\\
g_i\in\mathcal D_r,\,\,\,  Z(g_i)=\{z_i\}\subset Z(\Psi),\,\,z_i\in\mathbb C^+,\label{E:fact2}\\
 z_i\notin  Z(\psi_i),\label{E:fact3}
\end{gather}
for $i=1,\,2$. Then 
\[\textit{$\left(\mathcal L_\lambda g_1\right)g_1^{-1}$,  $\left(\mathcal L_\lambda \psi_2\right)\psi_2^{-1}$  are independent of $\lambda$.}
\]Where $\mathcal L_\lambda=\partial_y-\lambda\partial_x$.
\end{lemma} 
\begin{proof} The $\lambda$-independence of $\left(\mathcal L_\lambda g_1\right)g_1^{-1}$ has been proved by Theorem 6.5 in \cite{DT07}. We can use the same method to show the $\lambda$-independence of $\left(\mathcal L_\lambda \psi_2\right)\psi_2^{-1}$. That is, by (\ref{E:fact1}) and a direct computation, we have 
\begin{equation}
g_2^{-1}\left[\left(\mathcal L_\lambda \Psi\right)\Psi^{-1}-\left(\mathcal L_\lambda g_2\right)g_2^{-1}\right]g_2=\left(\mathcal L_\lambda \psi_2\right)\psi_2^{-1}.\label{E:rewrite}
\end{equation}
Note (\ref{E:fact2}), (\ref{E:fact3}) imply that $\left(\mathcal L_\lambda \psi_2\right)\psi_2^{-1}$ is holomorphic at $z$, $\bar z$.  Using (\ref{E:fact2}), (\ref{E:Laxini}) and the left side of (\ref{E:rewrite}),  one has that $\left(\mathcal L_\lambda \psi_2\right)\psi_2^{-1}$ is holomorphic outside $\{z,\,\bar z\}$. Moreover, $\left(\mathcal L_\lambda \psi_2\right)\psi_2^{-1}$ is bounded at $\infty$ by  $g_2\in\mathcal D_r$, (\ref{E:Laxini}), and the left side of (\ref{E:rewrite}). By the Liouville's theorem,  $\left(\mathcal L_\lambda \psi_2\right)\psi_2^{-1}$ is independent of $\lambda$. \end{proof}

\textit{\emph{Proof of \textbf{Theorem \ref{T:cauchy2}}:}} Since (\ref{E:chiral}) is the compatibility condition of (\ref{E:Lax01}), (\ref{E:Lax02}). To prove Theorem \ref{T:cauchy2}, it is sufficient to show the following theorem.

\begin{theorem}\label{T:cauchymix}
For $Q_0\in {\mathbb P}_{\infty,k,1}$, $k\ge 7$, if $Z=Z(\Psi_0)$ is finite and $Z\cap\mathbb R=\phi$, 
then  $\exists\Psi(x,y,t,\lambda)$, $Q(x,y,t)$ such that 
\begin{eqnarray}
&&(\partial_y-\lambda\partial_x)\Psi(x,y,t,\lambda)=\left(\partial_x Q(x,y,t)\right)\Psi(x,y,t,\lambda), \label{E:Lax032}\\
&&(\partial_t-\lambda\partial_y)\Psi(x,y,t,\lambda)=\left(\partial_y Q(x,y,t)\right)\Psi(x,y,t,\lambda), \label{E:Lax042}\\
&&Q(x,y,0)=Q_0(x,y).\nonumber
\end{eqnarray}
Moreover,  for $i+j+h\le N$,  and $\lambda\in \mathbb C\backslash \left(\mathbb R\cup Z\right)$ fixed,
\begin{eqnarray}
&&\textit{$\partial_x^i\partial_y^j\partial_t^h\Psi\in L_\infty$,}\label{E:diff6}\\
&&\textit{$\partial_x^i\partial_y^j\partial_t^h\left(\Psi-1\right)\to 0$ as $x$, or $y$, or $t\to\infty$.}\label{E:diff7}
\end{eqnarray}
\end{theorem}
\begin{proof}
By Lemma \ref{L:decom}, one has $\Psi_0=\tilde f_0g_0=\tilde g _0f_0$ with $f_0,\,\tilde f_0\in \mathcal D_c$, $g_0$, $\tilde g_0\in\mathcal D_r$ and (\ref{E:diff5}), (\ref{E:boundary5}) being satisfied. Since $Z=Z(\Psi_0)$ is finite and $Z\cap\mathbb R=\phi$. Successively multiplying $\Psi_0$ by factors as $\frac{\lambda-z_j}{\lambda-\bar z_j}$, $z_j\in Z(\Psi_0)\cap\mathbb C^-$, one obtains 
\begin{eqnarray}
&&(\partial_y-\lambda\partial_x)\Psi_0'(x,y,\lambda)=A'_0(x,y)\Psi_0'(x,y,\lambda), \label{E:mixLaxini}\\
&&\Psi_0'=\tilde f_0g_0'=\tilde g_0' f_0,\label{E:mixfact}\\
&&f_0,\,\tilde f_0\in \mathcal D_c,\,\,g'_0,\,\,\tilde g'_0\in\mathcal D_r,\label{E:mixfg}\\
&&Z(\Psi_0')=Z(g_0')=Z(\tilde g_0')\subset\mathbb C^+,\label{E:c+}
\end{eqnarray}
and for $i+j\le k-4$, and $\lambda\in \mathbb C\backslash \left(\mathbb R\cup Z\right)$,
\begin{gather}
\textit{$\partial_x^i\partial_y^j f_0,\,\,\partial_x^i\partial_y^j g'_0\in L_\infty$,}\label{E:diff9}\\
\textit{$\lim_{x^2+y^2\to\infty}\partial_x^i\partial_y^j \left(f_0-I\right)=\lim_{x^2+y^2\to\infty}\partial_x^i\partial_y^j \left(g'_0- CI\right)= 0$  }\label{E:boundary9}
\end{gather} Where $C$ is a scalar which is independent of $x$, $y$. Moreover, by (\ref{E:mixfg}), (\ref{E:c+}), and Theorem \ref{T:uhlenbeck}, one can factorize 
\begin{gather}
g'_0=g'_{0,1}\cdots g'_{0,k},\,\,\,\tilde g_0'=\tilde g'_{0,1}\cdots \tilde g'_{0,k}, \qquad g'_{0,j},\,\tilde g'_{0,j}\in\mathcal D_r, \label{E:mixlength}\\
\textsl{$Z(g'_{0,j})=\{\alpha_j\}\subset\mathbb C^+$, $Z(\tilde g'_{0,j})=\{\beta_j\}\subset\mathbb C^+$,}\label{E:ab1}\\
 \textit{$\alpha_i\neq\alpha_j$, $\beta_i\neq\beta_j$, if $i\neq j$.}\label{E:ab2}
\end{gather}

To prove the theorem, it is equivalent to showing the solvability of the following Cauchy problem
\begin{eqnarray}
&&(\partial_y-\lambda\partial_x)\Psi'(x,y,t,\lambda)=A'(x,y,t)\Psi'(x,y,t,\lambda), \label{E:mixLax032}\\
&&(\partial_t-\lambda\partial_y)\Psi'(x,y,t,\lambda)=B'(x,y,t)\Psi'(x,y,t,\lambda), \label{E:mixLax042}\\
&&\Psi'(x,y,0,\lambda)=\Psi'_0(x,y,\lambda),\,\,\,A'(x,y,0)=A'_0(x,y),\label{E:mixinitial2}
\end{eqnarray}
as well  for $i+j+h\le k-4$, and $\lambda\in \mathbb C\backslash \left(\mathbb R\cup Z\right)$ fixed,
\begin{gather}
\textit{$\partial_x^i\partial_y^j\partial_t^h \Psi'\in L_\infty$}\label{E:diff10}\\
\textit{$\partial_x^i\partial_y^j\partial_t^h\left(\Psi'-CI\right)\to 0$ as $|x|$, or $|y|$, or $|t|\to\infty$.}\label{E:boundary10}
\end{gather}
Since (\ref{E:Lax032}), (\ref{E:Lax042}) follow by computing the compatibility conditions of (\ref{E:mixLax032}), (\ref{E:mixLax042}). We are going to justify the solvability by inducing on the length $k$ (the multiplicity) in  (\ref{E:mixlength}).

For $k=1$, by (\ref{E:mixLaxini}), (\ref{E:mixfact}), Lemma \ref{L:tail}, there exist $A'_{r,0}(x,y)$, $A_{c,0}(x,y)$ such that
\begin{eqnarray}
&&(\partial_y-\lambda\partial_x)g'_0(x,y,\lambda)=A'_{r,0}(x,y)g'_0(x,y,\lambda),\label{E:mixLaxiniL}\\
&&(\partial_y-\lambda\partial_x)f_0(x,y,\lambda)=A_{c,0}(x,y)f_0(x,y,\lambda).\label{E:mixLaxiniR}
\end{eqnarray}
By (\ref{E:mixLaxiniL}), (\ref{E:mixfg}), (\ref{E:diff9}), (\ref{E:boundary9}) and Theorem \ref{T:cauchydiscrete}, one can find 
\begin{eqnarray}
&&(\partial_y-\lambda\partial_x)g'(x,y,t,\lambda)=A'_{r}(x,y,t)g'(x,y,t,\lambda), \label{E:mixLax032L}\\
&&(\partial_t-\lambda\partial_y)g'(x,y,t,\lambda)=B'_{r}(x,y,t)g'(x,y,t,\lambda), \label{E:mixLax042L}\\
&&g'(x,y,0,\lambda)=g'_0(x,y,\lambda),\,\,\,A'_{r}(x,y,0)=A'_{r,0}(x,y),\label{E:mixinitial2L}\\
&&g'(x,y,t,\cdot)\in\mathcal D_r,\,\,\,Z(g'(x,y,t,\cdot))=\{\alpha_1\},\label{E:g'}
\end{eqnarray}
and for $i+j+h\le k-4$, $\lambda\in \mathbb C\backslash \{z\}$ fixed,
\begin{gather}
\textit{$\partial_x^i\partial_y^j \partial_t^h g'\in L_\infty$ }\label{E:diff11}\\
\textit{$\lim_{x^2+y^2+t^2\to\infty}\partial_x^i\partial_y^j \partial_t^h\left(g'- CI\right)=0$.}\label{E:boundary11}
\end{gather}
By (\ref{E:mixLaxiniR}), (\ref{E:mixfg}), (\ref{E:diff9}), (\ref{E:boundary9}),  Lemma \ref{L:decom}, and Theorem \ref{T:cauchy}, one can find 
\begin{eqnarray}
&&(\partial_y-\lambda\partial_x)f(x,y,t,\lambda)=A_{c}(x,y,t)f(x,y,t,\lambda), \label{E:mixLax032R}\\
&&(\partial_t-\lambda\partial_y)f(x,y,t,\lambda)=B_{c}(x,y,t)f(x,y,t,\lambda),\label{E:mixLax042R}\\
&&f(x,y,0,\lambda)=f_0(x,y,\lambda),\,\,\,A_{c}(x,y,0)=A_{c,0}(x,y),\label{E:mixinitial2R}\\
&&f(x,y,t,\cdot)\in\mathcal D_c,\label{E:f}
\end{eqnarray}
and for $i+j+h\le k-4$, $\lambda\in \mathbb C\backslash \mathbb R$,
\begin{gather}
\textit{$\partial_x^i\partial_y^j \partial_t^h f\in L_\infty$,}\label{E:diff12}\\
\textit{$\lim_{x^2+y^2+t^2\to\infty}\partial_x^i\partial_y^j \partial_t^h\left(f- I\right)=0$.}\label{E:boundary12}
\end{gather}
Therefore, applying Theorem \ref{T:backlund}, we can find a unique 
\begin{equation}
\Psi'(x,y,t,\lambda)=\tilde f g'=\tilde g' f \label{E:pfg}
\end{equation}
such that $\Psi'$ satisfies (\ref{E:mixLax032}), (\ref{E:mixLax042}), (\ref{E:diff10}), (\ref{E:boundary10}), 
and $Z(\tilde g'(x,y,t,\cdot))=\{\alpha_1\}$. 

One the other hand, by (\ref{E:mixfact}), (\ref{E:mixfg}),  (\ref{E:mixLaxiniL}), (\ref{E:mixLaxiniR}), 
and Theorem \ref{T:backlundinitial}, we have $\Psi_0(x,y,\lambda)$, $\tilde f_0(x,y,\lambda)$, $\tilde g'_0(x,y,\lambda)$ are the unique functions which satisfiy (\ref{E:mixLaxini}). Since $g'(x,y,0,\lambda)=g'_0(x,y,\lambda)$, $f(x,y,0,\lambda)=f_0(x,y,\lambda)$,  we derive $\Psi'(x,y,0,\lambda)=\Psi'_0(x,y,\lambda)$. As a result, $A'(x,y,0)=A'_0(x,y)$ by (\ref{E:mixLaxini}), (\ref{E:mixLax032}) and we prove the theorem in case of $k=1$.

For $k>1$, we define 
\begin{eqnarray}
\tilde f_{k-1,0}&=&\tilde f_0g'_{0,1}\cdots g'_{0,k-1},\label{E:mixfactk1}\\
f_{k-1,0}&=&\tilde g'_{0,2}\cdots \tilde g'_{0,k}f_0.\label{E:mixfactk2}
\end{eqnarray}
Then (\ref{E:mixfact}), (\ref{E:mixlength}), (\ref{E:mixfactk1}), (\ref{E:mixfactk2}) imply
\begin{eqnarray}
&&\Psi_0'=\tilde f_{k-1,0}g_{0,k}'=\tilde g_{0,1}' f_{k-1,0},\label{E:mixfactk}\\
&&Z(g_{0,k}')=\{\alpha\}\subset\mathbb C^+,\,\,\alpha\notin Z(\tilde f_{k-1,0})\subset\mathbb C^+,\label{E:phiL}\\
&&Z(\tilde g_{0,1}')=\{\beta\}\subset\mathbb C^+,\,\,\beta\notin Z( f_{k-1,0})\subset\mathbb C^+.\label{E:phiR}
\end{eqnarray}
So 
by (\ref{E:mixLaxini}), (\ref{E:mixfactk})-(\ref{E:phiR}), and Lemma \ref{L:tail}, there exist $A''_{L,0}(x,y)$, $A''_{R,0}(x,y)$ such that
\begin{eqnarray*}
&&(\partial_y-\lambda\partial_x)g_{0,k}'(x,y,\lambda)=A''_{r,0}(x,y)g_{0,k}'(x,y,\lambda),\label{E:mixLaxiniL1}\\
&&(\partial_y-\lambda\partial_x)f_{k-1,0}(x,y,\lambda)=A''_{c,0}(x,y)f_{k-1,0}(x,y,\lambda).\label{E:mixLaxiniR1}
\end{eqnarray*}

Hence one can prove the theorem in case of $k>1$ by repeating the previous arguments except using the induction hypothesis to obtain the solution of the Cauchy problem
\begin{eqnarray*}
&&(\partial_y-\lambda\partial_x)f_{k-1}(x,y,t,\lambda)=A''_{c}(x,y,t)f_{k-1}(x,y,t,\lambda), \label{E:mixLax032R1}\\
&&(\partial_t-\lambda\partial_y)f_{k-1}(x,y,t,\lambda)=B''_{c}(x,y,t)f_{k-1}(x,y,t,\lambda),\label{E:mixLax042R1}\\
&&f_{k-1}(x,y,0,\lambda)=f_{k-1,0}(x,y,\lambda),\,\,\,A''_{c}(x,y,0)=A''_{c,0}(x,y).\label{E:mixinitial2R1}
\end{eqnarray*}
\end{proof} 

We conclude this report by a brief remark on examples of  $Q_0\in \mathbb P_{\infty,k,1}$, $k\ge 7$, and   $Z(\Psi_0)<\infty$. 
First  we let
$v(x,y,\lambda)=v(x+\lambda y,\lambda)$ satisfy
\[\det (v)=1,\quad
v=v^*>0,\quad v-1\in \mathcal{S},
\]
and for $\forall i$, $j,\,h\ge 0$, 
\begin{eqnarray*}
 &&\partial_x^i\partial_y^j\partial_\lambda^h\left(v-1\right)\in L_2({\rm R},\,d\lambda)\cap L_1({\rm R},\,d\lambda)\textit{ uniformly}, \\
 &&\partial_x^i\partial_y^j\partial_\lambda^h\left(v-1\right)\to 0 \,\,\textit{ in $L_2({\rm R},\,d\lambda)$ uniformly, as}\,\,|x|,\,|y|\to\infty.
\end{eqnarray*}
Here $ \mathcal{S}$ is the space of Schwartz functions. We can solve the inverse problem and obtain $\psi_0\in \mathcal{S}$ by the argument in proving Theorem \ref{T:LNexistence}. Note here we need to use the reality condition $v=v^*>0$ to show the global solvability. 
Moreover, by using the fomula $q_0(x,y)=\frac 1{2\pi i}\int_{\rm R}\psi_{0,-}(v-1)d\xi$, one obtains that $q_0$ is Schwartz and possesses purely continuous scattering data.

Therefore, one can successively apply the theory of adding (or substracting) 1-soliton in \cite{DTU} to construct $Q_0$ such that $Q_0\in \mathbb P_{\infty,k,1}$, $k\ge 7$, and $Z(\Psi_0)<\infty$. 

\vskip.2in

\noindent{\textbf{Acknowledgements}}

I would like to express my gratitude to  Chuu-Lian Terng who has taught me so much about the algebraic and geometric structure of the  integrable systems. 

\end{document}